\documentclass{llncs}[10point]

\usepackage[lncs,amsthm,autoqed]{pamas}
\usepackage{arydshln}
\usepackage{wrapfig}

\usepackage{enumitem,lipsum}
\usepackage[linesnumbered,ruled]{algorithm2e}
\usepackage{paralist}

\newcommand\eat[1]{}

\usepackage{tikz}
\usetikzlibrary{arrows}

\usepackage{mathrsfs}

\usepackage{calrsfs}

\newcommand{\Pref}[1][]{
  \ifthenelse{\equal{#1}{}}{\mathrel \succsim}{\mathop{\succsim_{#1}}}
}   
\newcommand{\sPref}[1][]{ 
  \ifthenelse{\equal{#1}{}}{\mathrel \succ}{\mathop{\succ_{#1}}}
}   
\newcommand{\Indiff}[1][]{   
  \ifthenelse{\equal{#1}{}}{\mathrel \sim}{\mathop{\sim_{#1}}}
}
\newcommand{\prefset}[1][]{\ifthenelse{\equal{#1}{}}{\mathcal{\succsim}}{\mathcal{\succsim}_{#1}}}

\newlength{\wordlength}

\title{Equilibria in Sequential Allocation}

\author{Haris Aziz\inst{1}  and Paul Goldberg\inst{2} and Toby Walsh\inst{1}}
  \institute{%
  Data61, CSIRO and UNSW Australia\\
    \email{\{haris.aziz,toby.walsh\}@data61.csiro.au}\\
    \and
    University of Oxford, UK\\
        \email{paul.goldberg@cs.ox.ac.uk}
  }

\begin{document}

\maketitle

\begin{abstract}
Sequential allocation is a simple mechanism for sharing multiple indivisible items.
We study strategic behavior in sequential allocation. In particular,
we consider Nash dynamics, as well as the computation and Pareto optimality  of pure equilibria, and Stackelberg strategies.
We first demonstrate that, even for two agents, better responses can cycle. 
We then present a linear-time algorithm that returns a profile (which we call the ``bluff profile'') that is in pure Nash equilibrium.
Interestingly, the outcome of the bluff profile is the same as that of the truthful profile and the profile is in pure Nash equilibrium for \emph{all} cardinal utilities consistent with the ordinal preferences. 
We show that the outcome of the bluff profile is Pareto optimal with respect to pairwise comparisons.
In contrast, we show that an assignment may not be Pareto optimal with respect to pairwise comparisons even if it is a result of a preference profile that is in pure Nash equilibrium for all utilities consistent with ordinal preferences. 
Finally, we present a dynamic program to compute an optimal Stackelberg strategy for two agents, where the second agent has a constant number of distinct values for the items.
\end{abstract}

\section{Introduction}

A simple but popular mechanism to allocate indivisible items is {\em sequential allocation} 
\citep{AWX15b,BoLa11a,BrSt79a,BrTa96a,KNW13a,KoCh71a,LeSt12a}. 
Sequential allocation is used, for example, by the Harvard Business School to allocate courses to students~\citep{BuCa12a} as well as multi-million dollar sports drafts~\citep{BrSt79a}.
In a sequential allocation mechanism, 
a {\em picking sequence} specifies the turns of the agents. 
For example, for sequence 1212, agents 1 and 2 alternate with agent 1 taking the first turn.   Agents report their preferences over the items. Then the items are allocated to the agents in the following manner. 
In each turn, the agent in that turn is given the most preferred item that has not yet been allocated.
In this paper we focus on the ``direct revelation'' version where agents submit
their complete rankings at the same time (and are committed to them),
as opposed to the ``extensive form'' version where agents take turns choosing
and are only committed to items chosen previously.
Sequential allocation is an ordinal mechanism since the outcome only depends on the ordinal preferences of agents over items. 
Although the agents are asked to report ordinal preferences, we will assume a standard assumption in the literature that agents have underlying additive utilities for the items.

It has long been known that sequential allocation is not strategy-proof
when agents do not have consecutive turns. 
An agent may not pick their most preferred item
remaining if they expect this item to remain 
till a later turn. Instead, the agent
may pick a slightly less preferred item that
they would not otherwise get. Of course, this
requires reasoning  about how the
agents may behave strategically at the 
same time. 
Since the sequential allocation mechanism is not
strategy-proof, how precisely
should agents behave?
There has already been some
work on strategic behavior in the setting
where sequential allocation is viewed as a repeated game. \citet{KoCh71a} presented a linear-time algorithm to compute a subgame perfect Nash equilibrium (SPNE) when there are two agents and the picking sequence is alternating ($121212\ldots$).
The result was generalized to the case of any sequence~\cite{KNWX13a}.
\citet{BrSt79a} stated that ``no algorithm is known which will
produce optimal play more efficiently than by checking many branches
of the game tree.'' Recently, it was proved that there can be an
exponential number of subgame perfect Nash equilibria and finding
even one of them is PSPACE-hard for an unbounded number of
agents~\cite{KNWX13a}. 

However, it is also natural to view sequential
allocation as a {\em one shot} game rather than a repeated
game. At the Harvard Business School, 
students submit a single ranked list of courses to 
a central organization that runs the sequential
allocation mechanism on these fixed preferences. 
This is essentially then a one shot game. 
This suggests considering the more general 
solution concept
of pure Nash equilibrium rather
than that of subgame perfect Nash equilibrium. 
In this paper, \emph{we will view  sequential allocation as a one shot strategic game in
  which the possible actions of the agents are possible ordinal
  preferences over the items, and the agents know each others' true ordinal
  preferences, as well as the picking sequence}. 
Surprisingly no algorithm to date has yet 
been proposed in the literature for
efficiently computing a pure Nash equilibrium
(PNE). We therefore propose a simple
linear time method to compute a PNE even for
an unbounded number of agents. 
We also consider Pareto optimality of pure Nash equilibria.
This issue is similar to previous work on price of anarchy/stability of equilibria in other strategic domains. 
Finally, we consider Stackelberg strategies in sequential allocation
where an agent announces the preference he or she intends to report.

\paragraph{Results}

We study the computational problems of
finding the equilibria of sequential allocation when
viewed as a one shot game.
No algorithm to date has
been proposed in the literature for
efficiently computing a pure Nash equilibrium
(PNE) of sequential allocation. 

One general method to compute a PNE is
to compute a sequence of better responses.
Indeed, for any finite potential game, this
is guaranteed to find a PNE. 
We first show better responses need not converge to a pure Nash equilibrium. Even for two agents, better responses can cycle. 
Instead, we propose a simple
linear time method to compute the preference
profile of a PNE even for
an unbounded number of agents. 
We refer to the output of this 
algorithm as the \emph{bluff profile}. 
Interestingly, the allocation generated by
the bluff profile is the same as that of the truthful profile,
and this profile is in equilibrium for \emph{all} cardinal utilities consistent with the ordinal preferences. 
The fact that this equilibrium can be computed
in linear time is perhaps a little surprising because computing 
just a single best response with the sequential allocation mechanism
has been recently shown to be NP-hard~\cite{ABLM16a}.
In addition, computing a subgame perfect Nash equilibrium
of the repeated game is PSPACE-hard~\cite{KNWX13a},
and this is a PNE of the one shot game. 
Our result that there exists a
linear-time algorithm to compute a PNE profile in
the one shot game also contrasts with the
fact that computing a PNE profile is NP-hard under the related
probabilistic serial (PS) random assignment mechanism for fair
division of indivisible goods~\cite{AGM+15d}.

We also consider Pareto optimality and
other fairness properties of the pure Nash equilibria (Section~\ref{sec:pareto}).
This is in line with work on the price of anarchy/stability 
of equilibria in other strategic domains. 
We show that the outcome of the bluff profile is Pareto optimal with respect to pairwise comparisons (defined in Section~\ref{sec:pareto}).
Hence, in sequential allocation, pure Nash equilibrium is not incompatible with ordinal Pareto optimality.
On the other hand, we also prove that an assignment may not be Pareto optimal with respect to pairwise comparisons even if it is a result of a preference profile that is PNE for all utilities consistent with ordinal preferences. 

Finally, in Section~\ref{sec:commitment} we show that an agent may have an
advantage from committing and declaring his preference and that committing to the truthful report may not be optimal. For 2 players we present a polynomial-time algorithm to compute an optimal strategy to commit to in the case that the other agent has a small number of utility values.

\section{Preliminaries}

We consider the setting in which we have $N=\{1,\ldots, n\}$ a set of agents, $O=\{o_1,\ldots, o_m\}$ a set of items, and the preference profile $\succ=(\succ_1,\ldots, \succ_n)$ specifies for each agent $i$ his complete, strict, and transitive preference $\succ_i$ over $O$. 

Each agent may additionally express a cardinal utility function $u_i$ consistent with $\succ_i$: $u_i(o)> u_i(o') \text{ iff } o\succ_i o'.$ We will assume that each item is positively valued, i.e, $u_i(o)>0$ for all $i\in N$ and $o\in O$. The set of all utility functions consistent with $\succ_i$ is denoted by $\mathcal{U}(\succ_i)$. We will denote by $\mathcal{U}(\succ)$ the set of all utility profiles $u=(u_1,\ldots, u_n)$ such that $u_i\in \mathcal{U}(\succ_i)$ for each $i\in N$. 
When we consider agents' valuations according to their cardinal utilities,
we will assume additivity, that is $u_i(O')=\sum_{o\in O'}u_i(o)$ for each $i\in N$ and $O'\subseteq O$.

An \emph{assignment} is an allocation of items to agents, represented as an
$n\times m$ matrix $[p(i)(o_j)]_{\substack{1\leq i\leq n, 1\leq j\leq m}}$
such that for all $i\in N$, and $o_j\in O$, $p(i)(o_j)\in \{0,1\}$;  and for all $j\in \{1,\ldots, m\}$, $\sum_{i\in N}p(i)(o_j)= 1$. 
An agent $i$ gets item $o_j$ if and only if $p(i)(o_j)= 1$. Each row $p(i)=(p(i)(o_1),\ldots, p(i)(o_m))$ represents the \emph{allocation} of agent $i$.

We will also present the cardinal utilities in matrix form.
A utility matrix $U$ is an $n\times m$ matrix $[U(i)(j)]_{\substack{1\leq i\leq n, 1\leq j\leq m}}$ such that for all $i\in N$, and $j\in O$, the entry $U(i)(j)$ in the $i$-th row and  $j$-th column is $u_i(o_j)$.
We say that utilities are \emph{lexicographic} if for each agent $i\in N$, $o\in O$, {$u_i(o)>\sum_{o'\prec_i o}u_i(o')$}. By $S\succ_i T$, we will mean $u_i(S)>u_i(T)$.

\begin{example}
Consider the setting in which $N=\{1,2\}$, $O=\{o_1,o_2,o_3,o_4\}$, the preferences of agents are
  \begin{align*}
    &1:\quad o_1, o_2, o_3, o_4&
    &2:\quad o_1, o_3, o_2, o_4
  \end{align*}
Then for the picking sequence $1221$, agent $1$ gets $\{o_1,o_4\}$ while $2$ gets $\{o_2,o_3\}$. The assignment resulting from sequential allocation (SA) can be represented as follows.
\[
SA(\succ_1,\succ_2) = \begin{pmatrix}
    1&0&0&1\\
    0&1&1&0
\end{pmatrix}.
\]
The allocation of agent $1$ is denoted by $SA(\succ_1,\succ_2)(1)$.
\end{example}

For a reported preference profile $(\succ_1',\ldots, \succ_n')$, an agent $i$'s \emph{best response} is a preference report $\succ_i''$ that maximizes utility 
$u_i(SA(\succ_i'',\succ_{-i}')(i))$.
We say that a reported preference profile $(\succ_1',\ldots, \succ_n')$ is in \emph{pure Nash equilibrium (PNE)} if no agent $i$ can report a preference $\succ_i''$ such that $u_i(SA(\succ_i'',\succ_{-i}')(i)) > u_i(SA(\succ')(i)).$

\section{Nash dynamics}\label{sec:nashdynamics}

Since we are interested in computing a PNE, a natural approach is to
simulate better responses and hope they converge.
For finite potential games, such an approach is
guaranteed to find a PNE. 
However, we show that even for two agents, computing better responses 
will not always terminate and is thus not a method that is guaranteed
to find a pure Nash equilibrium.


  %
   
\begin{theorem}
For two agents, better responses can cycle.
\end{theorem}

\begin{proof}
Let the sequence be the alternating one: $121212\ldots$.
The following 5 step sequence of better responses leads
to a cycle.

The ordinal preferences corresponding to the utility functions are as follows.
\begin{align*}
      \succ_1:\quad & o_3,o_4,o_5,o_6,o_9,o_{10},o_7,o_8,o_1,o_2  \\
      \succ_2:\quad & o_9,o_{10},o_5,o_6,o_7,o_8,o_1,o_2,o_3,o_4  
\end{align*}

It is sufficient to consider the agents having lexicographic utilities although the argument works for any utilities consistent with the ordinal preferences. 

    \smallskip
             \noindent
    This yields the following assignment and utilities at the start:

             \noindent

    \[
     SA(\succ_1,\succ_2) = \begin{pmatrix}
        1&0     & 1&1     &1&0     & 1&0    & 0&0     \\
        0&1     & 0&0     &0&1     & 0&1    & 1    &1 \\
     \end{pmatrix}
    \]

  
\noindent
In Step 1, Agent 1 misreports to increase his utility.
\begin{align*}
    \succ_1^1:\quad & o_5, o_6,  o_7, o_8,  o_3, o_4,  o_1, o_2,  o_9, o_{10}  \\
    \succ_2^1:\quad & o_9,  o_{10},  o_5, o_6,  o_7, o_8,  o_1, o_2, o_3, o_4
\end{align*}
    \[
    SA(\succ^1) = \begin{pmatrix}
        0&0     & 1&1     & 1&1         & 1&0        & 0&0     \\
        1&1     & 0&0     & 0&0        & 0&1        & 1    &1 \\
    \end{pmatrix}
    \]

%

       \noindent
       In Step 2, Agent 2 changes his report in response.
\noindent
        \begin{align*}
          \succ_1^2:\quad &o_5, o_6, o_7, o_8, o_3, o_4, o_1, o_2, o_9, o_{10}  \\
          \succ_2^2:\quad & o_5, o_6, o_7, o_8, o_9, o_{10}, o_1, o_2, o_3, o_4
         \end{align*}
    \[
             SA(\succ^2) = \begin{pmatrix}
                 1&0     & 1&1     & 1&0     & 1&0    & 0&0     \\
                 0&1     & 0&0     & 0&1     & 0&1    & 1&1     \\
             \end{pmatrix}
             \]


       \noindent
       In Step 3, Agent 1 changes his report in response.
            \begin{align*}
          \succ_1^3:\quad &o_5, o_6, o_9, o_{10}, o_3, o_4, o_1, o_2, o_7, o_8  \\
          \succ_2^3:\quad &  o_5, o_6, o_7, o_8, o_9, o_{10}, o_1, o_2, o_3, o_4
         \end{align*}
      \[
             SA(\succ^3) = \begin{pmatrix}
                 0&0     & 1&1     & 1&0     & 0&0        & 1&1     \\
                 1&1     & 0&0     & 0&1     & 1&1        & 0&0     \\
             \end{pmatrix}
             \]

       \noindent
       In Step 4, Agent 2 changes his report in response.
\begin{align*}
          \succ_1^4:\quad &o_5, o_6, o_9, o_{10}, o_3, o_4, o_1, o_2, o_7, o_8  \\
          \succ_2^4:\quad &  o_9, o_{10}, o_5, o_6, o_7, o_8, o_1, o_2, o_3, o_4
         \end{align*}
   \[
    SA(\succ^4) = \begin{pmatrix}
        1&0     & 1&1     & 1&1         & 0&0        & 0&0     \\
        0&1     & 0&0     & 0&0        & 1&1        & 1&1     \\
    \end{pmatrix}
    \]

       \noindent
       In Step 5, Agent 1 changes his report in response.
\begin{align*}
          \succ_1^5:\quad &o_5, o_6, o_7, o_8, o_3, o_4, o_1, o_2, o_9, o_{10}  \\
          \succ_2^5:\quad &  o_9, o_{10}, o_5, o_6, o_7, o_8, o_1, o_2, o_3, o_4
         \end{align*}  \[
             SA(\succ^5) = \begin{pmatrix}
                 0&0     & 1&1     & 1&1         & 1&0        & 0&0     \\
                 1&1     & 0&0     & 0&0        & 0&1        & 1&1     \\
             \end{pmatrix}
             \]


\noindent
Since $\succ^1=\succ^5$, we have cycled.
\end{proof}

\section{The Bluff Profile}\label{sec:bluff}

In this section, we outline a  linear-time algorithm to compute a pure Nash equilibrium preference profile. 
Surprisingly, we will show that the preference profile constructed is in pure Nash equilibrium for \emph{all} utilities consistent with the ordinal preferences. 



\begin{quote}
\emph{Simulate sequential allocation with the truthful preferences. Set the preferences of each agent to the order in which the items are picked when simulating sequential allocation under truthful preferences.}
\end{quote}
  
We refer to the profile constructed as the bluff profile since the idea behind the profile is that an agent wants to get the most preferred item immediately because if he does not, some other agent will take it.
We observe the following characteristics of the bluff profile. 
  
%
%

%
%
%
%
%
%

\begin{lemma}\label{lemma:bluff-charac}
In the bluff profile,
  \begin{inparaenum}[(i)]
    \item all agents have the same preferences;
    \item the order in which items are picked is the same as the order in which items are picked under the truthful profile; and
    \item the allocations of agents are the same as in the truthful profile.
  \end{inparaenum}
\end{lemma}


We show that the bluff profile is in pure Nash equilibrium if the utilities are lexicographic. 

\begin{lemma}\label{lemma:bluff-lex}
The bluff profile is in pure Nash equilibrium if the utilities are lexicographic. 
\end{lemma}

\begin{proof}
We prove by induction on the number of picks that no agent has an incentive to pick some other item when his turn comes which means that he picks the same item that he picks in the bluff profile which is also the most preferred item among the available items. This is equivalent to proving that no agent has an incentive to change his report from that in the bluff profile.

For the base case, let us consider the first agent who takes the first turn.
If he does not take his most-preferred item, the next agent will take it.
Since utilities are lexicographic, the first agent gains most by getting his most-preferred item.
Regarding the other agents, they are not disadvantaged by placing that item first in their preferences lists,
since it is taken by the first agent. It does not affect their ability to express their preferences
amongst the remaining items.

Similarly, let us assume that agents in the first $k$ turns did not have an incentive to misreport and pick some item other than the most preferred available item. Then we show that agent $j$ in the $k+1$-st turn does not have an incentive to change his report. Note that the item that $j$ picks according to the bluff profile is his most preferred item $o$ amongst those still available. This is because the order in which items are picked and allocations that are made exactly coincide with the truthful profile. Now if $j$ does not make the consistent pick, he will not be able to recover the loss of not getting $o$ because the utilities are lexicographic. Again, for the other agents (who do not get $o$) it does not disadvantage them to put $o$ in the $k+1$-st place in their preference list.
\end{proof}

We now prove the following lemma. 

\begin{lemma}\label{lemma:nto2}
Consider a profile in which all agents in $N\setminus \{i\}$ report the same preferences. Then agent $i$'s best response results in the same allocation for him for all utilities consistent with the ordinal preferences.
\end{lemma}

\begin{proof}
When all agents in $N\setminus \{i\}$ report the same preferences, then for agent $i$, from the perspective of agent $i$, all the turns of agents in $N\setminus \{i\}$ can be replaced by a single agent, representative of $N\setminus \{i\}$ who has the same preferences as agents in $N\setminus \{i\}$. Thus, computing a best response for agent $i$ when all agents in $N\setminus \{i\}$ report the same preferences is equivalent to computing a best response for agent $i$ when there is only one other agent (with the same preference as the agents in $N\setminus \{i\}$) and each turn of agents in $N\setminus \{i\}$ is replaced by the representative agent.
When there is one other agent, \citep{BoLa14b} proved that the best response results in the same allocation for the agent for all utilities consistent with the ordinal preferences.\footnote{This argument does not work when the number of other agents is more than one and they have different preferences. It can be shown that for three or more agents, best responses need not result in the same allocation.}
\end{proof}

Combining Lemmas~\ref{lemma:bluff-lex} and \ref{lemma:nto2}, we are in a position to prove the following:

\begin{theorem}\label{thm:consistent}
The bluff profile is in pure Nash equilibrium under all utilities consistent with the ordinal preferences.
\end{theorem}

\begin{proof}
From Lemma~\ref{lemma:bluff-lex}, we know that the bluff profile is in pure Nash equilibrium if the utilities are lexicographic. From Lemma~\ref{lemma:bluff-charac}, we know that all agents have the same preferences in the bluff profile. This immediately implies that for any agent $i$, all agents in $N\setminus \{i\}$ report the same preferences. 
From Lemma~\ref{lemma:nto2}, each agent $i$'s best response to the bluff profile results in the same unique allocation for all utilities consistent with the ordinal preferences. This allocation should be the same as allocation achieved by $i$ when he reports the bluff preferences because they yield the best allocation under lexicographic utilities. Hence the bluff profile is in pure Nash equilibrium under all utilities consistent with the ordinal preferences.
\end{proof}

\section{The Crossout Profile}\label{sec:crossout}

Since sequential allocation can also be viewed as a perfect information extensive form game, it admits a SPNE (Subgame-Perfect Nash Equilibrium) and hence a pure Nash equilibrium for the \emph{game tree}.\footnote{For readers not familiar with extensive form games and Subgame-Perfect Nash Equilibrium), we refer them to \citep{LeSh08a}.}
Computing a SPNE of the game tree  is PSPACE-complete~\citep{KNWX13a}. On the other hand, the optimal play for the extensive form game can be computed in polynomial time for the case of two agents.
The strategy corresponding to the SPNE is to play so that the last
agent gets their least preferred item, the second from last, their
next least preferred item, and so on.
We first show that for the case of two agents, similar ideas can also be used to construct a PNE preference profile for the one-shot game.




%


We use the expression {\em crossout profile} to refer to the preference profile in which both agents have the preferences which are the same as the item picking ordering in the optimal play of perfect information extensive form game. The crossout preference profile can be computed as follows:

\begin{quote}
\emph{Reverse and then invert (exchange 1s with 2s and vice versa) the picking sequence. Reverse the preferences of the two agents. Find the order $L$ in which items are allocated to the agents according to the new picking sequence and preferences. Return reverse of $L$ as the preference of each agent.}
\end{quote}

We now show that the crossout profile is in PNE for certain utilities consistent with the ordinal preferences.
We say that utilities are \emph{upward lexicographic} if for each agent $i\in N$ and two allocations with equal number of items, if the agent prefers the allocation with the better least-preferred item and in case of equality, the one with a better second-least-preferred item and so on. Such a preference relation can be captured by cardinal utilities as follows.
If agent $i$ has ordinal preferences $o_1,o_2,\ldots, o_m$, then utilities are as follows: $u_i(o_j)=1 - (1/2^{m+1-j})$ for all $j\in \{1,\ldots, m\}$.


%
%
%
%

\begin{lemma}\label{th:up-lex}
For two agents and for upward lexicographic utilities, the crossout profile is in PNE.
\end{lemma}

\begin{proof}[Proof Sketch]
Consider agent $i\in \{1,2\}$ and denote by $-i$ the other agent.
Let $\pi$ be the sequence of turns of the agents so that $\pi(j)$ is the agent with the $j$-th turn.
Now if agent $-i=\pi(m)$ has the last turn, then in the first $m-1$ turns, whenever agent $i$'s turn comes, he has an option to get an item better than $i$'s least preferred item $o_m$. Hence $i$ can guarantee to not get his least preferred item $o_m$ and hence guarantee $-i$ to get $o_m$. This will always be the best response for agent $i$ if he has upward lexicographic utilities. Since $-i$ gets $o_m$ in any case, $-i$ may as well rank $o_m$ last, and use his higher slots to prioritise amongst the other items.
We can now consider a situation in which $o_m$ does not exist in $O$ and it is fixed as the least preferred item in both agents' preferences.  Then the same argument can be applied recursively. 
\end{proof}

Lemma~\ref{th:up-lex} can be used to prove the following theorem.

\begin{theorem}
For two agents and for all utilities consistent with the ordinal preferences, the crossout profile is in PNE.
\end{theorem}

\begin{proof}
From Lemma~\ref{th:up-lex}, we know that for two agents and for upward lexicographic utilities, the crossout profile is in PNE.
When there is one other agent, Bouveret and Lang \citep{BoLa14b} proved that the best response results in the same allocation for the agent under all utilities consistent with the ordinal preferences.
\end{proof}

Next we show that even for two agents, the outcome of a crossout profile may not be the same as the truthful assignment.

\begin{example}\label{example:SPNE-not-truthful}
Even for two agents, the outcome of the crossout profile (and hence the SPNE assignment) may not be the same as the truthful assignment.
Consider the sequence 1212 and profile:
\begin{align*}
  \succ_1:\quad &a,b,c,d&
  \succ_2:\quad & b,c,a,d
\end{align*}

The picking sequence obtained after reversing and inverting the picking sequence is again 1212. The modified preferences are as follows.
\begin{align*}
  \succ_1'':\quad &d,c,b,a&
  \succ_2'':\quad &d, a, c, b
\end{align*}

Under picking sequence 1212 and profile $\succ''$, the items are picked as follows: $d, a, c, b$. We reverse this ordering to obtain the following crossout profile:
\begin{align*}
  \succ_1'':\quad &b, c, a, d&
  \succ_2'':\quad &b, c, a, d.
\end{align*}

Under this profile and original picking sequence 1212, 1 gets $\{b,a\}$ and $2$ gets $\{c,d\}$.
Also note that the SPNE path is as follows: 1 gets $b$, 2 gets $c$, 1 gets $a$ and then 2 gets $d$.
In contrast, in the truthful assignment, 1 gets $a$ and $c$.
\end{example}

\section{Pareto Optimality of Pure Nash Equilibria}\label{sec:pareto}

We next consider the Pareto optimality of equilibria.
An allocation $S$ is \emph{at least as preferred  with respect to pairwise comparisons by a given agent $i$} as allocation $T$, if there exists an 
an injection $f$ from $T$ to $S$ such that for each item $o\in T$, $i$ prefers $f(o)$ at least as much as $o$. 
We note that an agent strictly prefers $S$ over $T$ with respect to pairwise comparisons if $S$ results from $T$ by a sequence of replacements of an item in $T$ with a strictly more preferred item.
Note that the pairwise comparison relation is transitive but not necessarily complete. We will focus on Pareto optimality with respect to pairwise comparisons. 

We first show there exists a PNE whose outcome is Pareto optimal with respect to pairwise comparisons. 
Hence, unlike some other games, Pareto optimality is not incompatible with Nash equilibria in sequential allocation.

\begin{theorem}
The outcome of the bluff profile is Pareto optimal with respect to pairwise comparisons. 
\end{theorem}
The argument is as follows. Since the outcome of the bluff profile is the same as the outcome of the truthful profile and since the outcome of each truthful profile is Pareto optimal with respect to pairwise comparisons~\citep{BrKi05a}, the outcome of the bluff profile is Pareto optimal with respect to pairwise comparisons as well. 
Although the argument for the theorem is simple, it shows the
following: if the truthful outcome satisfies some normative properties
such as envy-freeness or other fairness properties~\citep{AGMW15a}, we know that there exists at least one PNE which results in an assignment with the same normative properties.   
The theorem above is in sharp contrast with the result in \cite{BrSt79a}
that there exist utilities under which no SPNE assignment is Pareto optimal
with respect to pairwise comparisons.
Other relevant papers that deal with implementing Pareto optimal outcomes in other
settings include \citep{KMT15a} and \citep{Moul84b}.

Next, we show that there may exist a PNE  whose outcome is not Pareto optimal with respect to pairwise comparisons. The statement holds even if the PNE in question is in PNE with respect to all utilities consistent with the ordinal preferences!

\begin{theorem}
An assignment may not be Pareto optimal with respect to pairwise comparisons even if it is a result of a preference profile that is in PNE for all utilities consistent with ordinal preferences. 
\end{theorem}
\begin{proof}
Consider the preference profile:
  \begin{align*}
  \succ_1:&\quad a,b,c,d,e,f&
  \succ_2:&\quad e,f,b,a,d,c\\
  \succ_3:&\quad c,f,e,d,a,b
  \end{align*}

Let the sequence be 123123.   
Then the outcome of the truthful preference profile can be summarized as 
  \begin{align*}
  1:&\quad \{a,b\}&
  2:&\quad \{e,f\}&
  3:&\quad \{c,d\}
  \end{align*}

Consider the following profile $\succ'$:
  \begin{align*}
  \succ_1':&\quad c,f,a,b, d,e&
  \succ_2':&\quad b,a,e, c,d,f\\
  \succ_3':&\quad f,e,d, a, b, c
  \end{align*}
 
    Then the outcome of the profile $\succ'$ can be summarized as 
      \begin{align*}
      1:&\quad \{c,a\}&
      2:&\quad \{b,e\}&
      3:&\quad \{f,d\}
      \end{align*}

We argue that the profile $\succ'$ is in PNE. 
In his reported preference $\succ_1'$, agent $1$ gets $\{a,c\}$. The only better outcome agent $1$ can get is $\{a,b\}$.
If he goes for $a$ first, he does not get $c$ or $b$. If he goes for $b$
first, he does not get $a$. So agent $1$ plays his best response for
all utilities consistent with his ordinal preferences.  

In his reported preference $\succ_2'$, agent $2$ gets $\{e,b\}$.
The only better outcome agent $2$ can get is $\{e,f\}$.
Now agent $2$ in his best response will try to get $\{e,f\}$. If agent
$2$ tries to pick $f$ first, he will not get $e$. Hence agent $2$
plays his best response for all utilities consistent with his ordinal preferences. 

Finally, for agent $3$, he cannot get $c$. 
The best he can get is $\{e,f\}$. If $3$ goes for $e$ first, then he does not get $f$. If $3$ goes for $f$ first, he can only get $d$. The best he can get is $\{f,d\}$ so his reported preference is his best response for all utilities consistent with the ordinal preferences.  
\end{proof}
  

%

%

\section{Advantage of Commitment}\label{sec:commitment}

In prior work on strategic aspects of sequential allocation,
the focus has been on computing manipulations or equilibria.
We now consider another strategic aspect: Stackelberg strategies
to commit to in order to obtain outcomes that are better for the individual agent.
In this setting, agent 1 (the leader) announces a preference $R$
of all the items, and commits to selecting, whenever it is his turn,
the highest-ranked item in $R$ that is not yet taken.
The following example illustrates a leadership advantage.

\begin{example}
There are 2 agents and 4 items denoted $a$,$b$,$c$,$d$.
Suppose the agents choose items in order $1212$.
The  ordinal preferences are 
  \begin{align*}
  \succ_1:&\quad a, d, c, b&
  \succ_2:&\quad a, b, d, c
  \end{align*}

Then in an SPNE, agent $1$ takes item $a$, then agent
$2$ takes item $d$ (since agent $1$ will not take item $b$,
it is okay for agent $2$ to take $d$, ending up with $b$ and $d$).
Then agent 1 takes $c$.
Also if agent 1 reports the truth $a, d, c, b$, then agent $2$
is guaranteed to get $b$ so he can report $d, b, a, c$ and get
$\{b,d\}$ which means that $1$ gets $\{a,c\}$.

However, consider the case where agent $1$ is leader, and announces
the preference list $\succ_1':\quad a, b, d, c.$
Then agent 2 must use a preference list that results in agent 2 taking item $b$ first. Agent $1$ has a credible threat to take item $b$, if agent $2$ does not take it next (despite the fact that agent $1$ doesn't value item $b$). So, agent 1 gets items $a$ and $d$.
\end{example}

This raises the following question. For two agents, what is the complexity
of finding the best preference report for the leader,
assuming that the follower will best-respond.
Next, we consider an interesting special case in which the problem can be solved in polynomial time.
\begin{theorem}\label{thm:dp}
For $n=2$ and any fixed picking sequence,
there is an algorithm whose runtime is polynomial in the number of items $m$,
to compute an optimal Stackelberg strategy for agent $1$ when
agent $2$ has a constant number of distinct values for items.
\end{theorem}
We make the assumption, standard in the study of optimal
Stackelberg strategies, that if agent 2 has more than one
best response, then agent 1 breaks the tie in his (agent 1's)
favour.
Let $k$ (constant) be the number of distinct values that agent
2 has for items. Agent 1 has to identify a ranking of the items such that if
agent  2 best-responds, agent 1's total value is maximised.


It is convenient to proceed by solving the following slight generalisation
of the problem.
Given a picking sequence $P$ (a sequence of 1's and 2's of length $m$),
we add a parameter $\ell$, where $\ell$ is at most the number of 2's
in $P$, and agent 2 may receive only $\ell$ items.
We make the following observation:


\begin{observation}\label{obs:tokens}
Suppose (for picking sequence $P$) agent 2 is allowed to receive $\ell$ items.
We can regard  agent 2's selection of items as working
as follows. Given agent 1's preference ranking $\succ_1$, agent 2 places a token on the $\ell$
items in $\succ_1$ whose positions correspond to the positions of 2 in P.
Agent 2 is allowed to move any token from any item $x$ to an item $x'$, provided that $x\succ_1 x'$,
subject to the constraint that tokens lie at distinct items.
Finally, items marked with tokens are the ones that agent 2 receives.
Agent 2 chooses the most valuable set that can be obtained in this way.
\end{observation}

\begin{proposition}\label{claim:rank}
We may assume that in an optimal (for agent 1) ranking of the items,
if items $x$ and $x'$ have the same value for  agent 2,
and  agent 1 values $x$ higher than $x'$, then  agent 1
ranks $x$ higher than $x'$.
\end{proposition}
\begin{proof}
We claim first that since $x$ and $x'$ have the same value to
 agent 2, then given any ranking by  agent 1, any best response
by  agent 2 can be modified to avoid an outcome where  agent 2
takes the higher-ranked of $\{x,x'\}$, but not the lower-ranked of
$\{x,x'\}$.

Noting Observation~\ref{obs:tokens},
if the higher-ranked of $\{x,x'\}$ has a token, but not
the lower-ranked of them, the token can be moved to the lower-ranked
of $\{x,x'\}$ without loss of utility to  agent 2.
If  agent 1 ranks the lower-valued of $\{x,x'\}$ higher in $\succ_1$,
they can be exchanged, and the new ranking (with the right best response
by  agent 2) is at least as good for  agent 1.
\end{proof}

\begin{paragraph}{Notation:}
{
Recall that $m$ denotes the number of items.
Let $S_1,\ldots,S_k$ be the partition of the items
into subsets that agent 2 values equally.
For $1\leq i\leq k$ let $m_i=|S_i|$.
Let $o_{i,j}\in S_i$ be the member of $S_i$ that has the $j$-th
highest value to  agent 1.
Let $S_i(j)\subseteq S_i$ be the set $\{o_1,\ldots,o_j\}$,
that is, the $j$ highest value (to agent 1) members of $S_i$.
Let $U(j_1,\ldots,j_k;\ell)$ be the highest utility that  agent 1
can get, assuming that items $S_1(j_1)\cup \cdots \cup S_k(j_k)$
are being shared, and  agent 2 is allowed to take $\ell$ of them,
where $\ell\leq m$. 
In words, we consider subsets of the $S_i$
obtained by taking the best items in $S_i$, and consider
various numbers of items that we limit  agent 2 may to receive.
}
\end{paragraph}

\begin{proof} (of Theorem~\ref{thm:dp})
If picking sequence $P$ contains $m'$ occurrences of ``2'',
we are interested in computing $U(m_1,\ldots,m_k;m')$
and its associated ranking.
We express the solution recursively be expressing
$U(j_1,\ldots,j_k;\ell)$ in terms of various values of
$U(j'_1,\ldots,j'_k;\ell')$, where $j'_i\leq j_i$ and $\ell'\leq \ell$,
and at least one inequality is strict.
Furthermore, for any values $j_1,\ldots,j_k,\ell$ we also evaluate
and remember agent 2's best response.
This can be seen to be achievable in polynomial time via dynamic programming,
since there are $O(m^{k+1})$ sets of values that can be taken
by these parameters.

We compute $U(j_1,\ldots,j_k;\ell)$ as follows.
Let $j=j_1+\ldots+j_k$ and assume that there are at least $\ell$
occurrences of ``2'' in the first $j$ entries of the picking
sequence.

By Proposition~\ref{claim:rank}, in an (agent 1)-optimal ranking
of items in $S_1(j_1)\cup\cdots\cup S_k(j_k)$, the lowest ranked item
must be one of $o_{1,j_1},\cdots,o_{k,j_k}$.
We consider two cases, according to whether or not agent 2
takes that lowest-ranked item.

Suppose agent 2 takes that item.
Then $U(j_1,\ldots,j_k;\ell)$ is given by:
\begin{equation}\label{eq:1}
\max_{i\in[k]} \Bigl( U(j_1,\ldots,j_{i-1},j_i-1,j_{i+1},\ldots,j_k;\ell-1) \Bigr)
\end{equation}

Alternatively agent 2 may fail to take that item, in which case
$U(j_1,\ldots,j_k;\ell)$ is given by:
\begin{equation}\label{eq:2}
\max_{i\in[k]} \Bigl( u_1(o_{i,j_i}) + U(j_1,\ldots,j_{i-1},j_i-1,j_{i+1},\ldots,j_k;\ell) \Bigr)
\end{equation}
where $u_1(o_{i,j_i})$ is agent 1's value for item $o_{i,j_i}$
lowest-ranked.

By way of explanation of~(\ref{eq:1}) and~(\ref{eq:2}),
in the case of~(\ref{eq:1}) where agent 2 takes $o_{i,j_i}$,
agent 1's utility will
be that of the optimal ranking of the other $j-1$ items
under the constraint that agent 2 only gets to take $\ell-1$
of them. If agent 2 does not take this lowest-ranked item $o_{i,j_i}$,
then~(\ref{eq:2}) gives agent 1's utility as his value 
$u_1(o_{i,j_i})$ for that item, plus the best outcome for agent 1
assuming agent 2 may take $\ell$ items from amongst the other
$j-1$ items.

In checking which case applies for a given choice of $i$ and
corresponding item $o_{i,j_i}$,
we check whether the optimal ranking of the other items for agent 2
taking $\ell-1$ of them, when extended to  agent 2's selection of
that additional item, is indeed a best-response for agent 2
given that he gets $\ell$ of all the items.
This can be done efficiently, since best responses can be efficiently
computed~\citep{BoLa14b}.
\end{proof}

\section{Conclusion}

Sequential allocation is a simple and frequently used mechanism
for resource allocation.  Its
strategic aspects have been formally studied
for the last forty years. To our surprise, some fundamental questions
have been unaddressed in the literature about sequential allocation
when viewed as an one shot game. This is despite the fact that in many
settings, it is essentially played as an one shot game. 
We have therefore studied in detail 
the pure Nash equilibrium of sequential allocation
mechanisms. 
We presented a number of results on Nash dynamics,
as well as on the computation of pure Nash equilibrium,
and the Pareto optimality of equilibria. In particular, 
we presented the first polynomial-time algorithm to compute a PNE that applies to all utilities consistent with the ordinal preferences. 
We have also explored some other new directions such as Stackelberg strategies that have so far not been examined in sequential allocation.


%


 %

\end{document}